\documentclass[journal]{IEEEtran}
\IEEEoverridecommandlockouts
\usepackage{cite}
\usepackage{amsmath,amssymb,amsfonts}
\usepackage{amsthm}
\usepackage{subfig}
\usepackage{algorithmic}
\usepackage{graphicx}
\usepackage{textcomp}
\usepackage{siunitx}
\usepackage{xcolor}
\def\BibTeX{{\rm B\kern-.05em{\sc i\kern-.025em b}\kern-.08em
    T\kern-.1667em\lower.7ex\hbox{E}\kern-.125emX}}

\begin{document}

\bstctlcite{BSTcontrol}

\title{Data Generation Method for Learning a Low-dimensional Safe Region in Safe Reinforcement Learning
}

\author{Zhehua~Zhou\textsuperscript{1},
        Ozgur S.~Oguz\textsuperscript{2},
        Yi Ren\textsuperscript{3},
        Marion~Leibold\textsuperscript{1}
        and~Martin~Buss\textsuperscript{1}
\thanks{\textsuperscript{1}Chair of Automatic Control Engineering, Technical University of Munich, Munich, Germany (e-mail: zhehua.zhou@tum.de; marion.leibold@tum.de; mb@tum.de).}
\thanks{\textsuperscript{2}Max Planck Institute for Intelligent Systems and University of Stuttgart, Stuttgart, Germany (e-mail: ozgur.oguz@ipvs.uni-stuttgart.de).}
\thanks{\textsuperscript{3}Tencent Robotics X Lab, Tencent, Shenzhen, China (e-mail: evanyren@tencent.com).}}


\maketitle

\begin{abstract}
Safe reinforcement learning aims to learn a control policy while ensuring that neither the system nor the environment gets damaged during the learning process.
For implementing safe reinforcement learning on highly nonlinear and high-dimensional dynamical systems, one possible approach is to find a low-dimensional safe region via data-driven feature extraction methods, which provides safety estimates to the learning algorithm.
As the reliability of the learned safety estimates is data-dependent, we investigate in this work how different training data will affect the safe reinforcement learning approach.
By balancing between the learning performance and the risk of being unsafe, a data generation method that combines two sampling methods is proposed to generate representative training data.
The performance of the method is demonstrated with a three-link inverted pendulum example.
\end{abstract}


\section{Introduction}
\label{sec.introduction}
Deep reinforcement learning (RL) approaches have demonstrated impressive achievements in various control tasks of dynamical systems, e.g., humanoid control~\cite{peng2017deeploco} or robotic manipulator control~\cite{levine2016end}.
However, most of the RL algorithms are currently applied only in simulations~\cite{duan2016benchmarking}, as during the exploration process for an optimal policy, the system may encounter an unsafe intermediate policy that is harmful to the system itself or to the environment.
Hence, for employing RL approaches on real-world dynamical systems, it is critical to include safety guarantees in the learning process.

Safe reinforcement learning (SRL) in dynamical systems with continuous action space has become a popular topic in recent researches~\cite{garcia2015comprehensive}.
For model-free scenarios, safety is usually achieved via solving a constraint satisfaction problem.
For example, constrained policy optimization~\cite{achiam2017constrained} adds a constraint to the learning process that the expected return of cost functions should be restricted within certain predefined limits.
Alternatively, introducing an additional risk term in the reward function, e.g. risk-sensitive RL~\cite{shen2014risk}, can also increase the safety of RL algorithms.
However, as no exact system model is considered in these approaches, there exists still a high probability that safety conditions are violated, especially in the early phase of the learning process.

When an approximation of the system model is available, more reliable safety guarantees can be realized by combining techniques from model-based nonlinear control with RL approaches.
For example in~\cite{perkins2002lyapunov,chow2018lyapunov}, Lyapunov functions are employed to construct a safe subregion of the state space, referred to as a safe region, such that by limiting the learning process within the safe region, safety conditions will never be violated.
However, finding suitable candidates of Lyapunov functions is challenging if the system dynamics is highly nonlinear and high-dimensional.
Besides, robust model predictive control can also be used to provide safety and stability guarantees to RL algorithms~\cite{ostafew2016robust,zanon2020safe}. 
However, the performance of such approaches in general highly relies on the accuracy of the used system model.

To relax requirements on the quality of the system model, recent research introduces data-driven methods for designing a SRL approach that is based on probabilistic safety estimates.
For example in~\cite{fisac2018general}, safety in learning is modelled as a differential game. 
By approximating unknown external disturbances with Gaussian Process models, a probabilistic safe region is computed via reachability analysis.
Similarly in~\cite{berkenkamp2016safe,berkenkamp2017safe}, a safe region is represented by the region of attraction (RoA) and is estimated through modeling the unknown part of system dynamics with a Gaussian Process model.
The exploration of RL algorithms is restricted in such a forward invariant safe region, such that safety is preserved as long as a corrective controller is applied when the system approaches the boundary of the safe region.
However, finding the safe region becomes difficult when the system dynamics is highly nonlinear and high-dimensional (referred to as complex dynamical systems), as in this case performing the reachability analysis or estimating the RoA via sampling are both computationally infeasible~\cite{fisac2019bridging}.

In order to implement SRL on complex dynamical systems, which are often the systems of interest for applying deep RL algorithms, we propose a SRL framework that is based on finding a low-dimensional representation of the safe region~\cite{zhou2020general}.
For each high-dimensional system state, a low-dimensional corresponding state is computed and is considered as the safety feature that predicts the safety of the system.
Then these low-dimensional states are used to construct a reduced-order safe region that approximates the original high-dimensional safe region in a probabilistic form.
However, determining a reliable low-dimensional representation of the safe region is challenging, especially when a thorough understanding about system dynamics is lacking.
To overcome this problem, we introduce a data-driven feature extraction method in~\cite{zhou2020learning} to construct a well-performed low-dimensional safe region.
It assumes that, although the exact full system dynamics might be unknown, a nominal system model that provides at least rough estimates about the system behaviour is available.
Then by collecting data about safety of different system states from the nominal system, a low-dimensional safety feature is derived via learning the probabilistic similarities between training data points.
The mismatch between the nominal and the real systems is later compensated through an online adaptation method.
However, as the learned safety feature is data-dependent, how to generate training data points that are most useful to the SRL framework is an important problem to be solved.

The contribution of this work is two-fold. 
First, we propose a data generation method that is able to generate representative training data points by considering their potential influence on the performance of the SRL framework.
The low-dimensional representation of the safe region learned from those data points achieves a satisfying balance between exploring the state space for finding an optimal policy and keeping the system safe.
Second, we investigate how different training data points will affect the reliability of the derived safety estimates.
Taking the used SRL framework as an example, we provide an insight about what could be a useful way to generate training data points when any other data-driven method is employed to predict the safety of a dynamical system.


\section{Preliminaries}
\label{sec.SRL}
In this section, we outline a SRL framework for complex dynamical systems that is proposed in~\cite{zhou2020general,zhou2020learning}.
The framework relies on the construction of a low-dimensional representation of the safe region, i.e., a subregion of the state space from where the system can be controlled back to a safe state, for estimating the safety of different system states.

\subsection{SRL based on RoA}
We consider the nonlinear control affine dynamical system with partially unknown dynamics as
\begin{equation}
    \dot{x} = f(x) + g(x)u + d(x)
\label{eq.system_real}
\end{equation}
where $x \in \mathcal{X} \subseteq \mathbb{R}^n$ is the $n$-dimensional system state and $u \in \mathcal{U} \subseteq \mathbb{R}^m$ is the $m$-dimensional control input to the system.
$d(x)$ represents the unknown part of the system dynamics.
For brevity, we refer to this system as the \textit{real} system in this work.

Similar as in~\cite{berkenkamp2016safe}, we assume that the origin of the real system~\eqref{eq.system_real} is a known safe state and is locally asymptotically stable under a given corrective controller $u = K(x)$.
A system state $x$ is said to be a safe state if the system can be controlled back to the origin by using the corrective controller $K(x)$ when starting from this state.
According to this, a safe region is defined by using the RoA $\mathcal{R}$ of the origin with respect to the corrective controller $K(x)$ as follows.

\newtheorem{definition}{Definition}
\begin{definition}
\label{def.safe_region}
A safe region $\mathcal{S}$ is a closed positive invariant subset of the RoA $\mathcal{R}$.
\end{definition}

As long as the system state $x$ is inside the safe region $\mathcal{S}$, the system can always be controlled back to a safe state, i.e., the origin, by applying the corrective controller $K(x)$.
Therefore, we categorize system states into safe (safety label $z=1$) and unsafe ($z=0$) classes with a labeling function $l(x):\mathcal{X} \rightarrow \mathcal{Z} = \{1,0\}$ as
\begin{equation}
l(x) = z = \begin{cases} 
    1, & \text{if } x \in \mathcal{S} \backslash \{ \partial \mathcal{S} \} \\
    0, & \text{else} 
\end{cases} 
\label{eq.label_real}
\end{equation}
where $\mathcal{S} \backslash \{ \partial \mathcal{S} \}$ is the interior of the safe region $\mathcal{S}$ that excludes the boundary $\partial \mathcal{S}$.
Apparently, if the safe region $\mathcal{S}$ is known, then an ideal SRL framework can be designed based on a supervisory control strategy that switches between a learning-based controller $\pi(x)$ and the corrective controller $K(x)$ as
\begin{equation}
    u = \begin{cases} 
    \pi(x), & \text{if } t < t^{*} \\
    K(x), & \text{else} 
\end{cases} 
\label{eq.original_supervisor}
\end{equation}
where $t^{*}$ is the first time that the system state $x$ is considered as unsafe, i.e., $l(x) = 0$.
In each learning trial, the system starts inside the safe region $\mathcal{S}$ with time $t = 0$ and the learning-based controller $\pi(x)$ is first applied.
For keeping the system safe, the supervisor~\eqref{eq.original_supervisor} activates the corrective controller $K(x)$ at time $t = t^{*}$.
After the safety recovery, the learning environment is reset and the next learning trial starts again with time $t = 0$.
However, in general it is computationally infeasible to calculate the safe region $\mathcal{S}$ directly for a complex dynamical system~\cite{ahmadi2019dsos,fisac2019bridging}.


\subsection{SRL with a Low-dimensional Representation of the Safe Region}

For implementing the SRL framework on complex dynamical systems, an approach that is based on estimating the safety with a low-dimensional representation of the safe region is introduced in~\cite{zhou2020general}.

By mapping each system state $x$ to a low-dimensional simplified state, i.e., the safety feature, $y \in \mathcal{Y} \subseteq \mathbb{R}^{n_y}, n_y \ll n$ with a state mapping $y = \Psi(x)$, the safety of system state $x$ is estimated through safety of the corresponding simplified state $y$ in a probabilistic form as
\begin{equation}
    \mathbb{P}( l(x)  = 1) = \Gamma(y)|_{y = \Psi(x)} \sim [0,1]
\label{eq.safety_gamma}
\end{equation}
where $\Gamma(y)$ is a safety assessment function defined over the simplified state space $\mathcal{Y}$. 
In~\cite{zhou2020general}, the state mapping $y = \Psi(x)$ and the safety assessment function $\Gamma(y)$ is determined by using a physically inspired model order reduction technique.

By using a predefined probability threshold $p_t$, a low-dimensional representation of the safe region, denoted as the simplified safe region $\mathcal{S}_y$, is thus given in the simplified state space $\mathcal{Y}$ as
\begin{equation}
\mathcal{S}_y = \{ y \in \mathcal{Y} \hspace{1mm} | \hspace{1mm} \Gamma(y) > p_t \}
\label{eq.safe_region_low}
\end{equation}
which approximates the high-dimensional safe region $\mathcal{S}$.
It leads to a SRL framework for complex dynamical systems that is based on the following modified supervisor 
\begin{equation}
     u = \begin{cases} 
    \pi(x), & \text{if } t < t^{'} \\
    K(x), & \text{else }
\end{cases} 
\label{eq.prob_supervisor}
\end{equation}
where $t^{'}$ is the first time point that the system state $x$ is predicted to be unsafe, i.e., $\Psi(x) = y \notin \mathcal{S}_y$.

\subsection{Data-driven Feature Extraction}

To overcome the limitation of physically inspired model order reduction, we employ a data-driven feature extraction method for identifying the simplified safe region $\mathcal{S}_y$ in~\cite{zhou2020learning}.
It assumes that, the available knowledge about the system dynamics formulates a \textit{nominal} system 
\begin{equation}
    \dot{x} = f(x) + g(x)u
\label{eq.system_nominal}
\end{equation}
Due to the highly nonlinear and high-dimensional dynamics, calculating the safe region of the nominal system, denoted as $\mathcal{S}_n$, is still computationally infeasible.
However, the safety of each individual system state $x$ of the nominal system can be examined directly by simulating the nominal system with respect to the corrective controller $K(x)$.
Hence, a dataset that reflects the safe region of the nominal system $\mathcal{S}_n$ is obtainable.
By expecting that the behavior of the nominal system will at least provide a prediction about the behaviour of the real system, a data-driven method is thus implemented to derive an initial estimate of the simplified safe region $\mathcal{S}_y$ for the real system.

We refer to the dataset obtained from the nominal system as the training dataset $\mathbb{D}_{\mathrm{tr}}$ with $|\mathbb{D}_{\mathrm{tr}}|= k$ data points. 
Each data point contains both the system state $x$ at which the corrective controller is initially activated and the corresponding safety label of this state.
By comparing the pairwise similarities between training data points, a method called t-Distributed Stochastic Neighbor Embedding (t-SNE)~\cite{maaten2008visualizing} is adopted to compute a realization of simplified states $\{y_1,\ldots, y_k\}$ that best represents the training dataset $\mathbb{D}_{\mathrm{tr}}$. 
These simplified states are then used to approximate the state mapping $y = \Psi(x)$ and the safety assessment function $\Gamma(y)$ for calculating an initial estimate of the simplified safe region $\mathcal{S}_y$.
See~\cite{zhou2020learning} for more details about learning with t-SNE.

\section{Training Data Generation Method}
\label{sec.data_generation}

The estimate of the simplified safe region $\mathcal{S}_y$ obtained from the aforementioned data-driven method provides a hypothesis $h(x):\mathcal{X}\rightarrow\mathcal{Z}$ for predicting the safety label of different real system states as
\begin{equation}
h(x) = z = \begin{cases} 
    1, & \text{if }  \Psi(x) = y \in \mathcal{S}_y \\
    0, & \text{if }  \Psi(x) = y \notin \mathcal{S}_y 
\end{cases} 
\label{eq.label_hypothesis}
\end{equation}
The reliability of the hypothesis $h(x)$ depends, on the one hand, on the magnitude of discrepancy between the nominal and the real systems. 
On the other hand, as the state mapping $y = \Psi(x)$ and the safety assessment function $\Gamma(y)$ are derived from the training dataset $\mathbb{D}_{\mathrm{tr}}$, the quality of training data points also affects the performance of the hypothesis $h(x)$.
Therefore, we first investigate in this section the influence of the choice of training data on the classification error and the performance of SRL framework.
Then we propose a data generation method that combines a uniform distribution and a multivariate normal distribution to generate a training dataset $\mathbb{D}_{\mathrm{tr}}$ that is most useful to the SRL framework.

\subsection{Training Data vs. Classification Error}

Predicting safety of real system states can be treated as a binary classification problem. 
Hence, we consider the classification error as the first criterion when generating the training dataset $\mathbb{D}_{\mathrm{tr}}$.

We assume that during the learning, all visited system states of the real system are drawn from an unknown distribution $\mathcal{D}$.
Meanwhile, the distribution of the nominal system states contained in the training dataset $\mathbb{D}_{\mathrm{tr}}$ is denoted as $\mathcal{D}_n$.
Then for the learned hypothesis $h(x)$, its classification error on the real system $\epsilon (h, l)$ (referred to as the generalization error) and on the nominal system $\epsilon_n (h, l_n)$ (referred to as the source error) are
\begin{align}
\epsilon (h, l) = \mathrm{E}_{x\sim \mathcal{D}} \left[ \mathbb{I} (h(x) \neq l(x)) \right]
\label{eq.classification_error_real}
\end{align}
\begin{align}
\epsilon_n (h, l_n) = \mathrm{E}_{x\sim \mathcal{D}_n} \left[ \mathbb{I} (h(x) \neq l_n(x)) \right]
\end{align}
which represent the probability that according to the distribution $\mathcal{D}$ or $\mathcal{D}_n$, the hypothesis $h(x)$ disagrees with the labeling function $l(x)$ given by the safe region $\mathcal{S}$, or the labeling function $l_n(x)$ given by the safe region of the nominal system $\mathcal{S}_n$, respectively.
By extending Theorem 1 given in~\cite{ben2010theory} based on the $\mathcal{H}\Delta\mathcal{H}$-divergence, the following theorem that bounds the generalization error holds for the hypothesis $h(x)$.

\newtheorem{theorem}{Theorem}
\begin{theorem}
\label{theo.classification_erro} 
The generalization error $\epsilon (h, l)$ of hypothesis $h(x)$ satisfies
\begin{equation}
    \epsilon (h, l) \leq  \epsilon_n (h, l_n) + \frac{1}{2}d_{\mathcal{H}\Delta \mathcal{H}}(\mathcal{D},\mathcal{D}_n) + \min\{\mathrm{E}_1, \mathrm{E}_2\}
\label{eq.theorem_error}
\end{equation}
where $d_{\mathcal{H}\Delta \mathcal{H}}$ is the $\mathcal{H}\Delta\mathcal{H}$-distance and we have $\mathrm{E}_1 = \mathrm{E}_{x\sim \mathcal{D}_n} \left[ \mathbb{I} (l(x) \neq l_n(x)) \right]$ , $\mathrm{E}_2 = \mathrm{E}_{x\sim \mathcal{D}} \left[ \mathbb{I} (l(x) \neq l_n(x)) \right]$.
\end{theorem}

\begin{proof}
See Appendix B. 
\end{proof}

The upper bound of the generalization error given in~\eqref{eq.theorem_error} contains three terms: the first term is the source error; the second term represents the divergence in distributions; the third term is the difference in labeling functions and cannot be changed, as it is affected only by the discrepancy between the nominal and the real systems.
Hence for achieving a low generalization error, Theorem~\ref{theo.classification_erro} suggests to move closer the two distributions $\mathcal{D}$ and $\mathcal{D}_n$ while keeping the source error small.
Based on this, it is motivated to generate the training dataset $\mathbb{D}_{\mathrm{tr}}$ by using an accurate estimate of the unknown distribution $\mathcal{D}$.

In~\cite{zhou2020learning}, the training dataset $\mathbb{D}_{\mathrm{tr}}$ is generated by sampling system states with a uniform distribution (UD) $\mathcal{D}_{\mathrm{ud}}$ among the entire state space.
However when controlling a dynamical system, the probability that a system state $x$ will be visited is affected by the system dynamics (see Section~\ref{sec.result_initial} and in particular Fig.~\ref{fig.data_gaussian} for an example).    
Hence, in general the UD $\mathcal{D}_{\mathrm{ud}}$ is expected not to be close to the real distribution $\mathcal{D}$.

Although the distribution $\mathcal{D}$ is unknown prior to the learning process on the real system, it can be approximated by simulating the nominal system.
To do this, we first set the initial state of the nominal system as the origin.  
Then, we control the nominal system with a random policy and record all system states observed in the system trajectory. 
Repeating this multiple times results in a dataset $\mathbb{X}$ of system states that reflects the probability that different states will be visited during control.
A multivariate normal distribution (MND) $\mathcal{D}_{\mathrm{mnd}}(\mu, \Sigma)$ is then fitted to the dataset $\mathbb{X}$ and is considered as an approximation of the distribution $\mathcal{D}$. 
Apparently, the accuracy of such an approximation is affected by the magnitude of discrepancy between the nominal and the real systems.

\subsection{Classification Error and SRL}

The motivation of using the MND $\mathcal{D}_{\mathrm{mnd}}$ for generating training data points is from the prospective of reducing the generalization error.
However, if the data generation is decided only based on the generalization error, the performance of the SRL framework might be affected due to the following reason.

The generalization error consists of two parts
\begin{eqnarray}
\epsilon (h, l) &=& \mathrm{E}_{x\sim \mathcal{D}} \left[ \mathbb{I} (h(x) = 1, l(x) = 0) \right] \nonumber \\
&+& \mathrm{E}_{x\sim \mathcal{D}} \left[ \mathbb{I} (h(x) = 0, l(x) = 1) \right]
\label{eq.classification_error_decom}
\end{eqnarray}
While the first error type (false positive) will cause unsafe behaviours of the dynamical system, the second error type (false negative) only means conservativeness in the SRL process.
The purpose of SRL is to find a satisfying learning-based policy $\pi(x)$ while keeping a high probability that the system is safe.
Hence, conservativeness is acceptable as long as a well-performed policy can be learned within the subregion of state space restricted by the supervisor.
In that regard, considering only the generalization error in the training data generation is likely to deteriorate the performance of the SRL framework, as reducing the conservativeness usually also means a higher chance of encountering an unsafe system behavior.

Therefore, for taking the performance of the SRL framework into consideration, we in general would like to keep a certain degree of conservativeness during the learning process.
This can be achieved by reducing our confidence in considering a system state $x$ as safe unless enough evidence is provided.
In that sense, using the UD $\mathcal{D}_{\mathrm{ud}}$ for generating the training dataset $\mathbb{D}_{\mathrm{tr}}$ is helpful.
The reason is that, compared to the MND $\mathcal{D}_{\mathrm{mnd}}$ that has a majority of data points being close to the origin, the training data points are now placed among the entire state space (see Section~\ref{sec.result_initial} for an example). 
Thus the proportion of safe data points is reduced.  
Note that, the underlying principle of using data-driven method for making safety predictions is to use known data points for estimating the safety of unseen data points. 
Therefore, the hypothesis $h(x)$ learned from the UD $\mathcal{D}_{\mathrm{ud}}$ tends to make an unsafe prediction, since it is less likely to find a nearby safe training data point.
As a result, although the UD $\mathcal{D}_{\mathrm{ud}}$ gives a higher generalization error, it preserves the conservativeness in the SRL framework, which then ensures a higher probability that the system is safe during the learning process.

\subsection{Combined Data Generation}
While conservativeness is able to results in a safer learning process, a satisfying policy that completes the given control task might not be found if the RL algorithm is overly restricted. 
Hence for achieving a good balance between the learning performance and the probability of being safe, we propose to divide the training dataset $\mathbb{D}_{\mathrm{tr}}$ into two parts
\begin{equation}
    \mathbb{D}_{\mathrm{tr}} = \mathbb{D}_{\mathrm{ud}} + \mathbb{D}_{\mathrm{mnd}}
\end{equation}
with $|\mathbb{D}_{\mathrm{ud}}| = \alpha k$, $|\mathbb{D}_{\mathrm{mnd}}| = (1 - \alpha) k$ and $0 \leq \alpha \leq 1$.
The sub-datasets $\mathbb{D}_{\mathrm{ud}}$ and $\mathbb{D}_{\mathrm{mnd}}$ are generated by using the UD $\mathcal{D}_{\mathrm{ud}}$ and the MND $\mathcal{D}_{\mathrm{mnd}}$, respectively.
The coefficient $\alpha$ determines the size of sub-datasets as well as the tendency of the SRL framework to perform exploration or to keep the safety.
If the unknown dynamics $d(x)$ is assumed to be small, or failure of the corrective controller $K(x)$ is considered as less critical, it is suggested to use a small value of $\alpha$ for ensuring a satisfying performance of the RL algorithm, where training data points are mostly sampled by using known knowledge about the system trajectories.
On the contrary, if it is more important to avoid unsafe behaviours, then a large value of $\alpha$ should be used to keep the conservativeness in learning, i.e., most training data points are drawn from the UD.


\section{Experimental Results}
\label{sec.result}

In this section, we examine the influence of the combined data generation method on the performance of the SRL framework with a three-link inverted pendulum example.

\begin{figure}[t]
    \sf
    
    \centering
    \includegraphics[width = .55\linewidth]{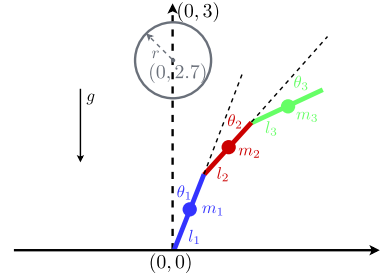}
    \caption{Three-link inverted pendulum with a target circle given in the Cartesian space. The connection point between the pendulum and the ground is the origin of the Cartesian coordinate system. When the pendulum is at the zero configuration, the end-effector point locates at $(0,3)$. The target circle has its centre at $(0,2.7)$ and its radius as $r = 0.3$.}
    \label{fig.inverted_pendulum}
\end{figure}

\begin{figure*}[!t]
\sf
\small
\centering
\subfloat[]{\includegraphics[width=.27\linewidth, height= 34mm]{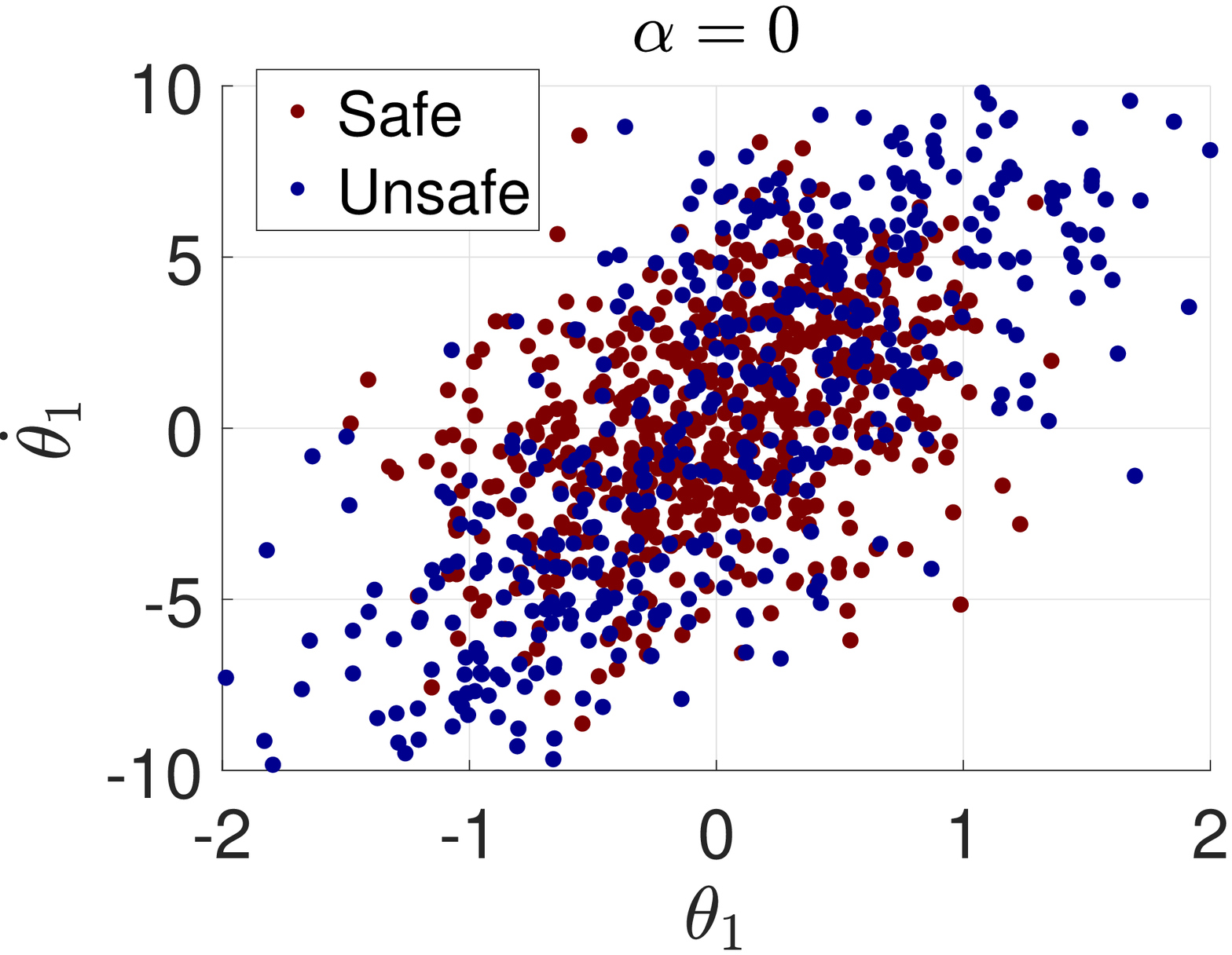}%
\label{fig.data_gaussian}}
\hfil
\subfloat[]{\includegraphics[width=.27\linewidth, height= 34mm]{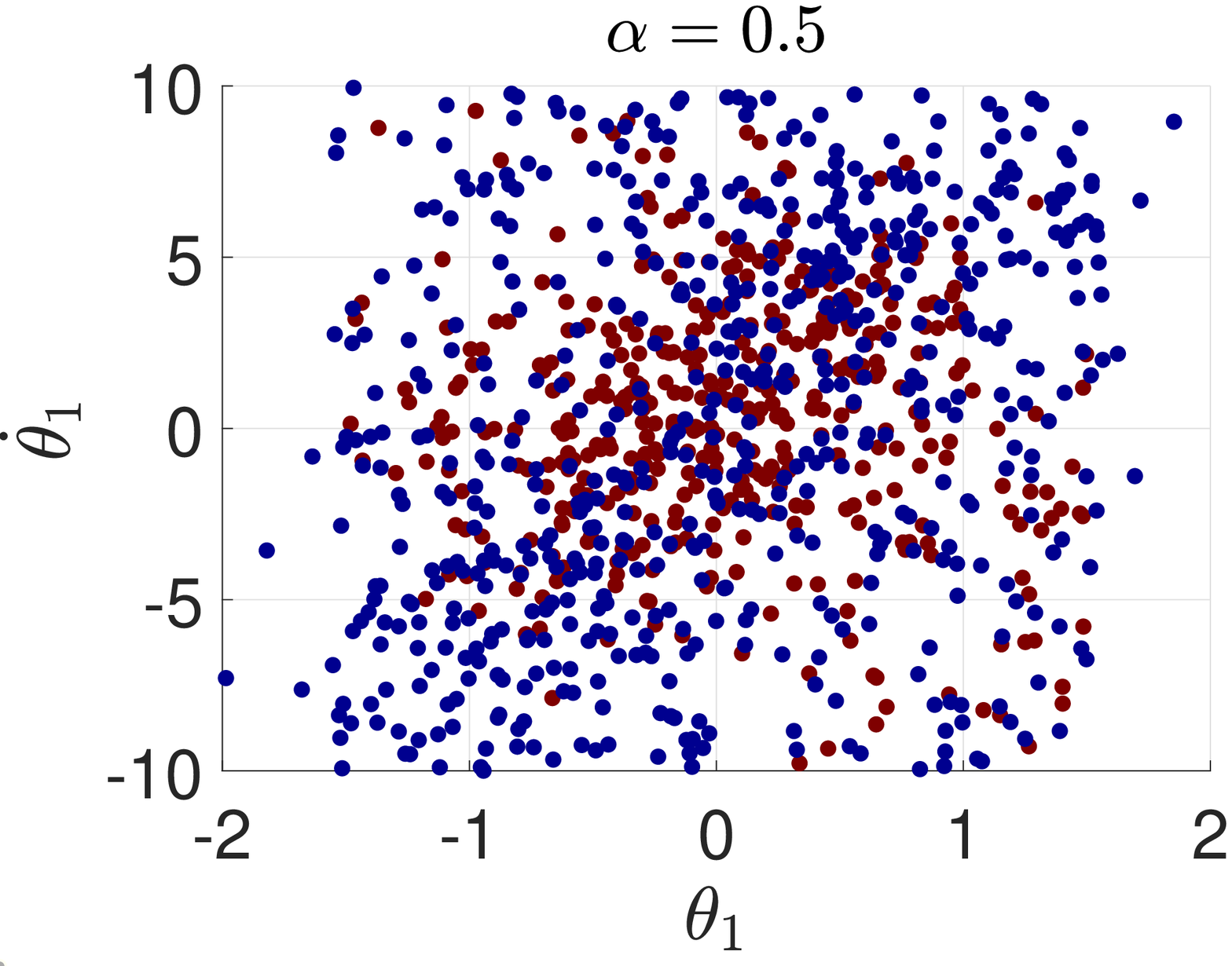}%
\label{fig.data_mix}}
\hfil
\subfloat[]{\includegraphics[width=.27\linewidth, height= 34mm]{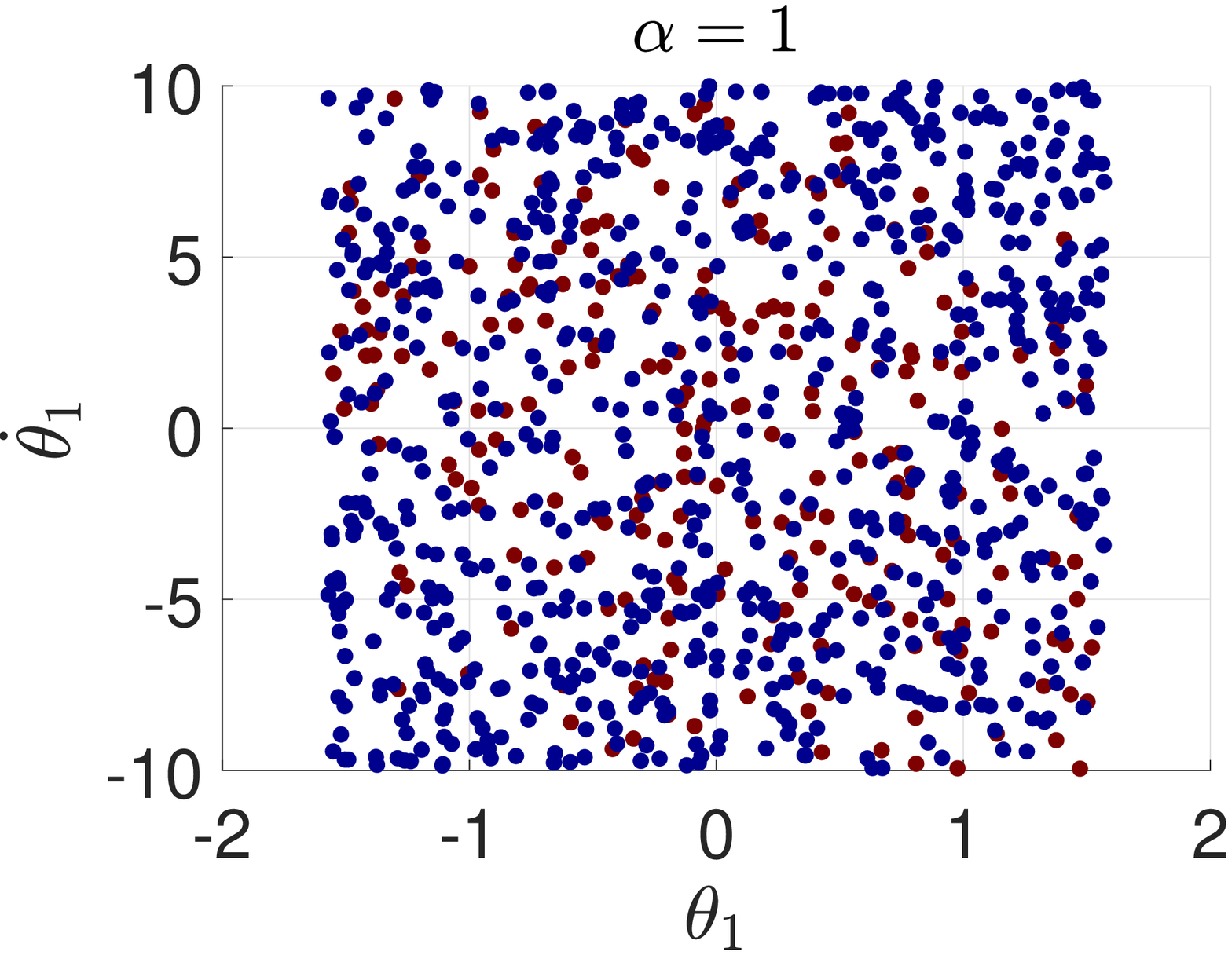}%
\label{fig.data_uniform}}
\hfil
\subfloat[]{\includegraphics[width=.27\linewidth, height= 34mm]{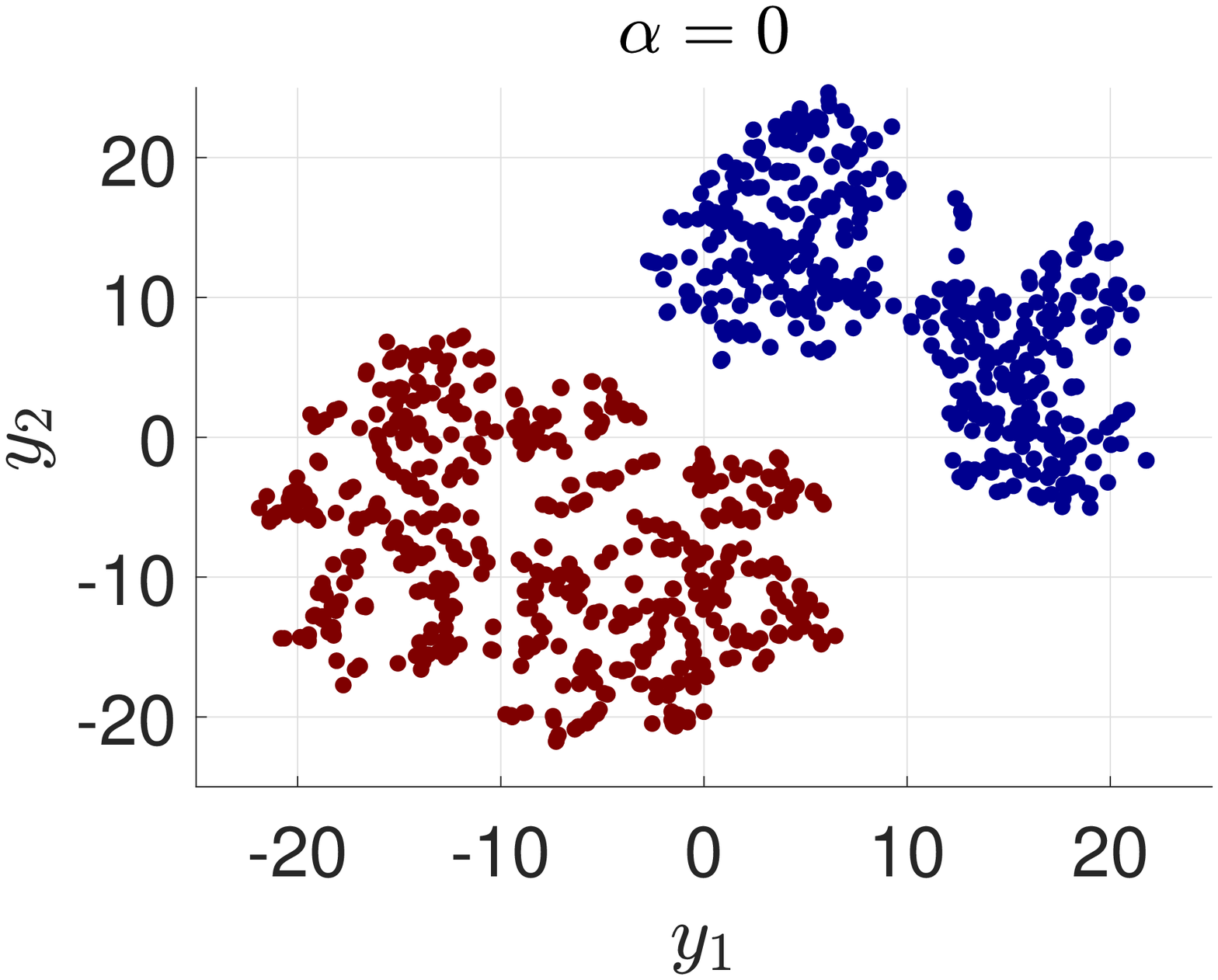}%
\label{fig.y_gaussian}}
\hfil
\subfloat[]{\includegraphics[width=.27\linewidth, height= 34mm]{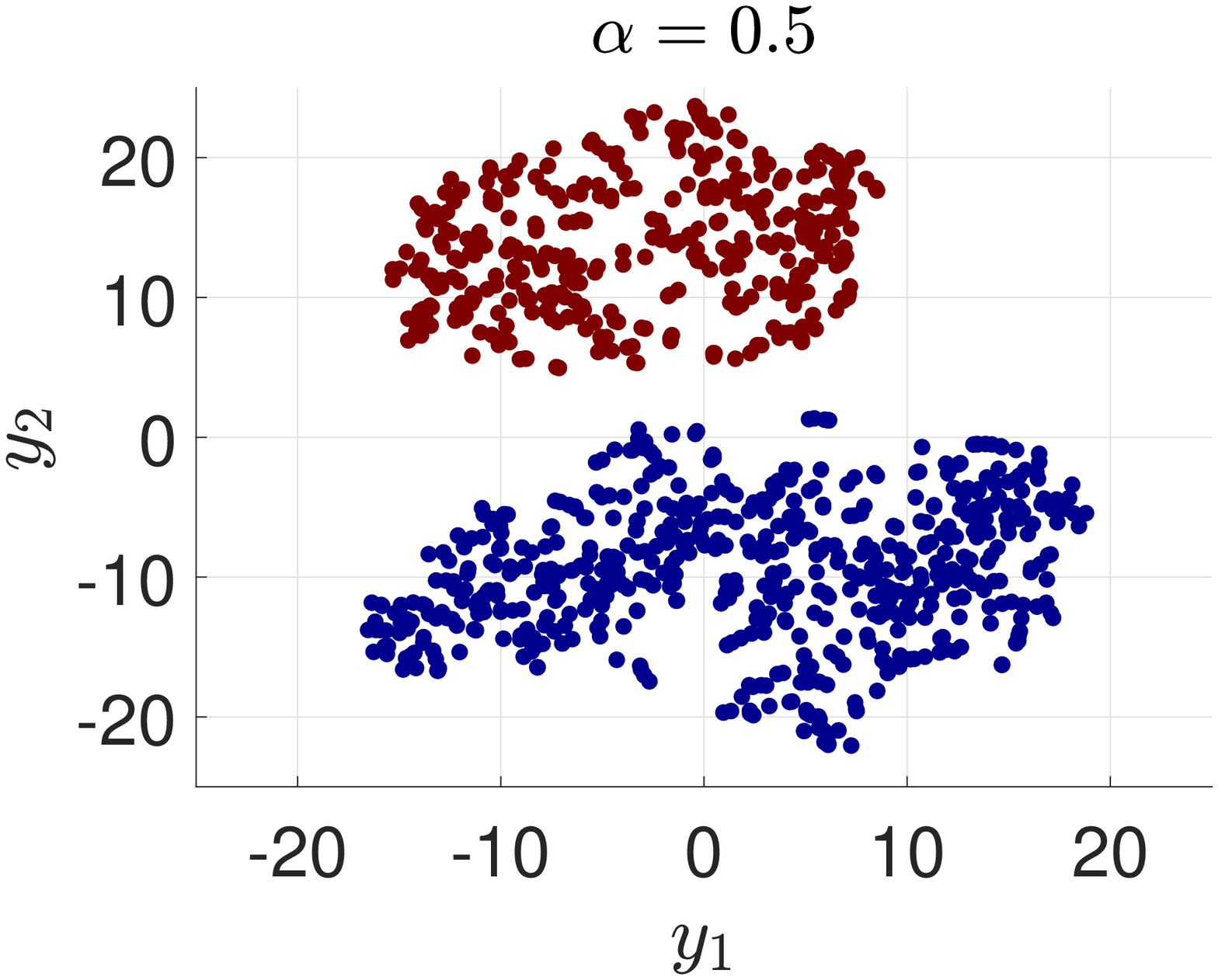}%
\label{fig.y_mix}}
\hfil
\subfloat[]{\includegraphics[width=.27\linewidth, height= 34mm]{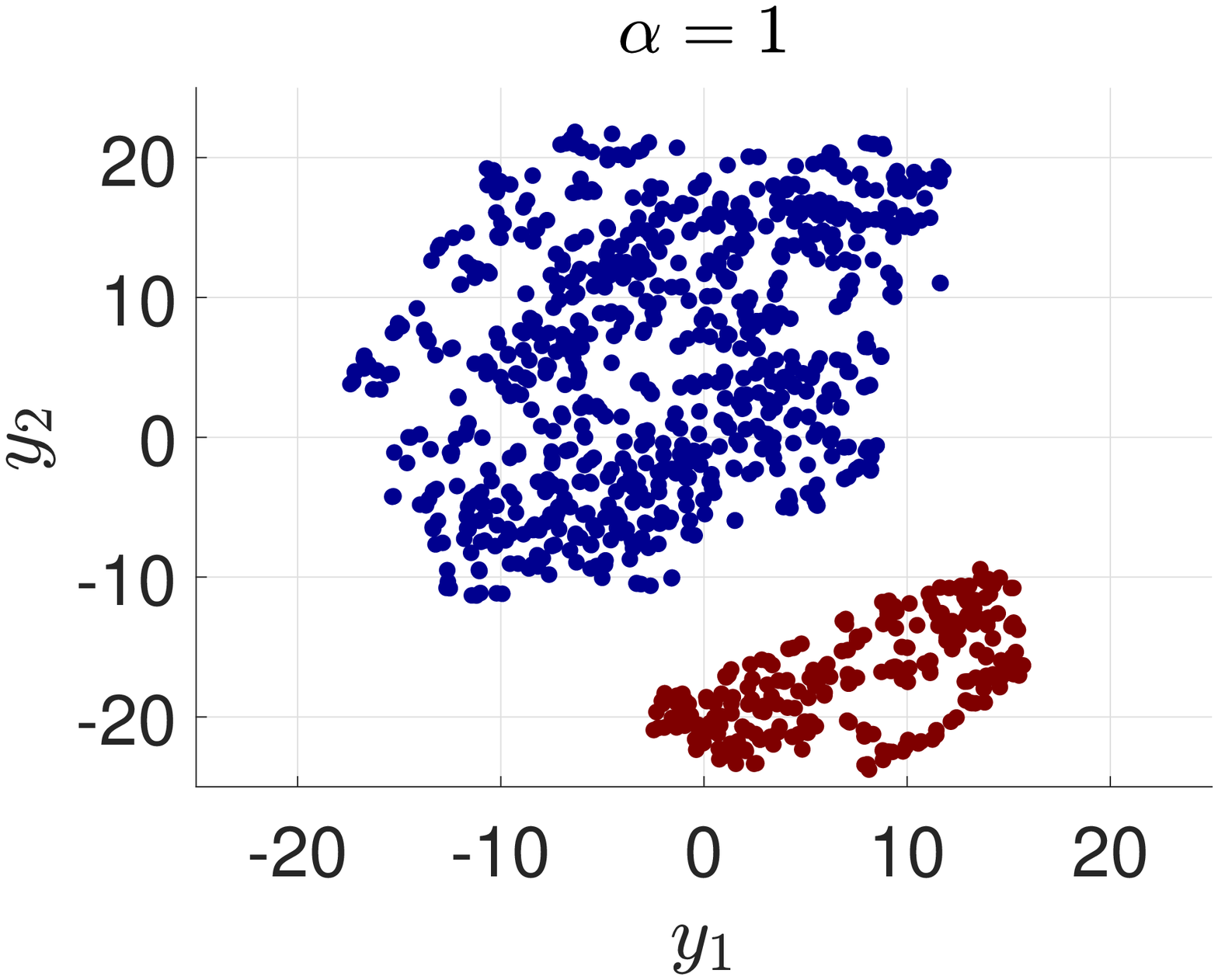}%
\label{fig.y_uni}}
\hfil
\subfloat[]{\includegraphics[width=.27\linewidth, height= 34mm]{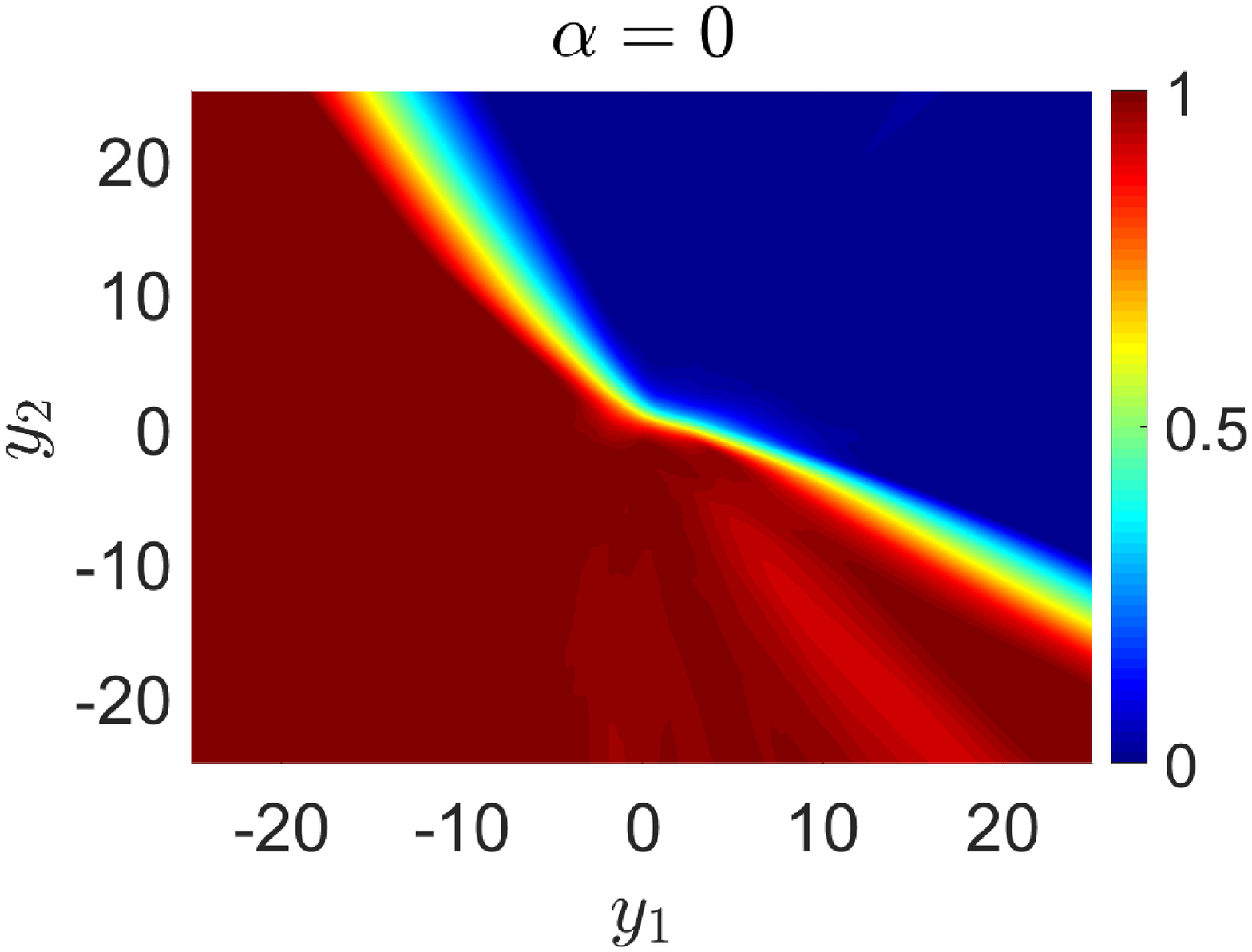}%
\label{fig.gamma_gaussian}}
\hfil
\subfloat[]{\includegraphics[width=.27\linewidth, height= 34mm]{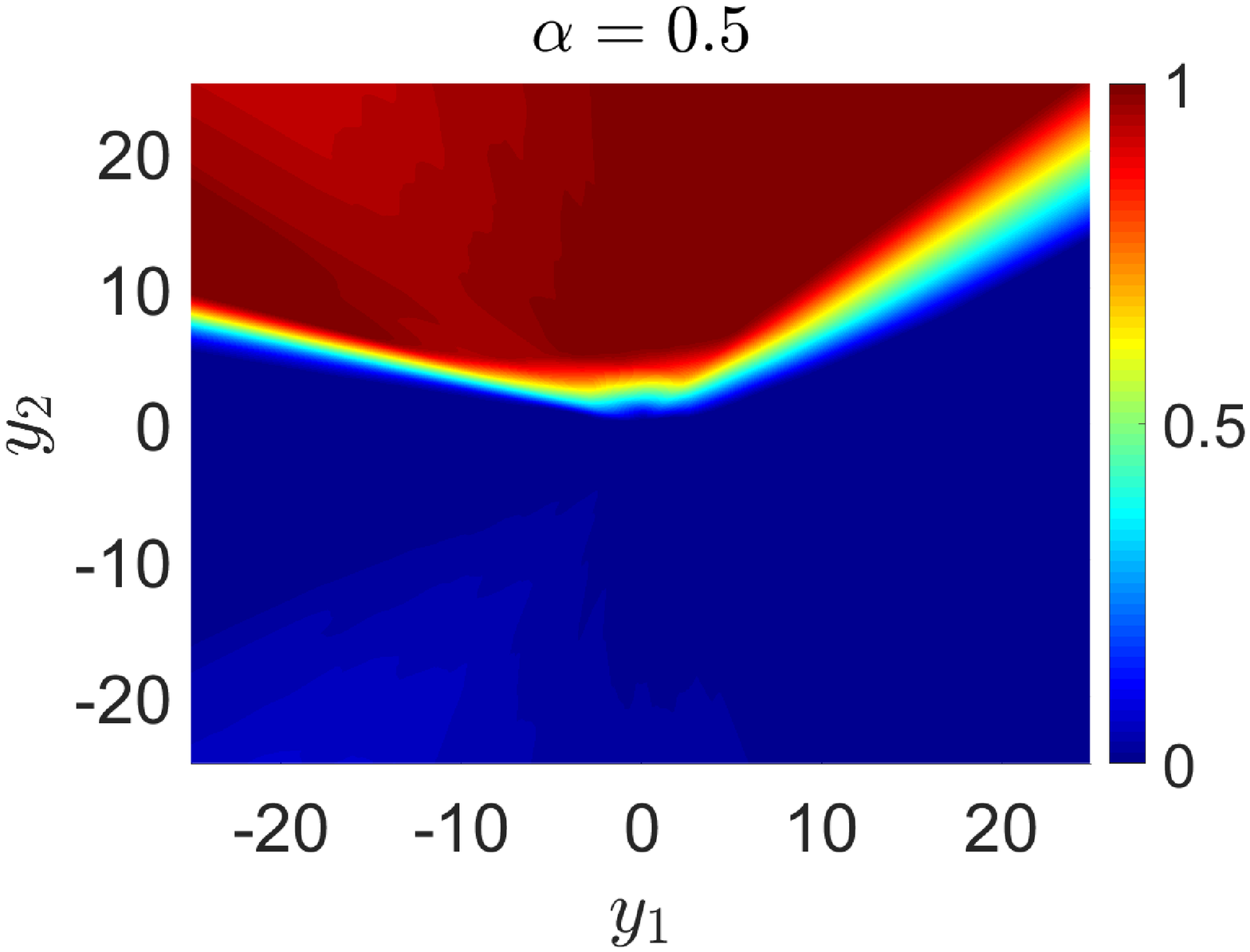}%
\label{fig.gamma_mix}}
\hfil
\subfloat[]{\includegraphics[width=.27\linewidth, height= 34mm]{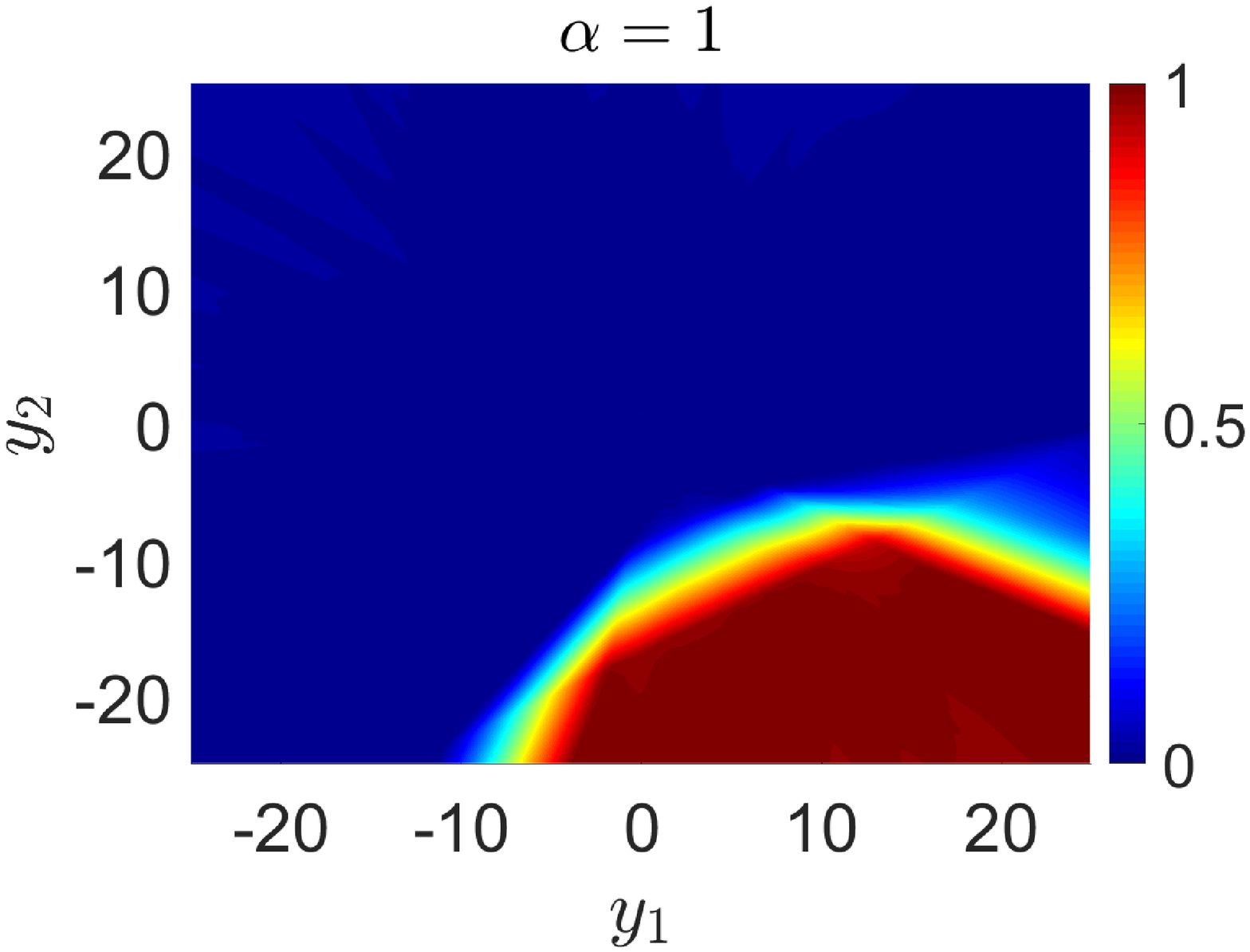}%
\label{fig.gamma_uni}}
\caption{(a)-(c) Distribution of joint angle and angular velocity of the first link in the sampled system states for $\alpha = 0$, $\alpha = 0.5$ and $\alpha = 1$, respectively. (d)-(f) The derived realization of simplified states. (g)-(i) The learned simplified safe region $\mathcal{S}_y$. The output of the safety assessment function $\Gamma(y)$ is represented by different colors. }
\label{fig.initial_estimate_results}
\end{figure*}

\subsection{Experimental Setup}

We consider a three-link inverted pendulum given as in Fig.~\ref{fig.inverted_pendulum}. 
The learning task is to find a control policy $\pi(x)$ that makes the end-effector point of the pendulum track a trajectory given as a circle in the Cartesian space with an angular velocity with respect to the centre of the circle as $\pi$ \si{\radian\per\sec}.
During the learning, we attempt to keep the system safe by preventing the first link from hitting the ground.

The system state is $6$-dimensional and consists of three joint angles and three joint angular velocities as $x = [\theta_1, \theta_2, \theta_3, \dot{\theta}_1, \dot{\theta}_2, \dot{\theta}_3 ]^T$.
The inputs $u = [u_1, u_2, u_3]^T$ are the torques applied on the three joints, where the maximal and minimal allowed torques are $u_{\mathrm{max}} = 100$ \si{\newton\metre} and $u_{\mathrm{min}} = -100$ \si{\newton\metre} for all three joints.
The lengths of the links are set to $l_1 = l_2 = l_3 = 1$ \si{\metre}.
We assume that the masses are concentrated on the centre of masses that are located at the middle point of each link.
For the nominal system, we consider the masses as $m_1 = m_2 = m_3 = 1$ \si{kg}.
The discrepancy between the nominal and the real systems is assumed to be caused by the mismatch in the masses of the first and the second links as $m_1 = m_2 = \Delta\cdot1$ \si{kg}. 
We use the Proximal Policy Optimization (PPO)~\cite{schulman2017proximal} algorithm as the learning-based controller, and the corrective controller $K(x)$ is a LQG controller that is derived from the nominal system. 
When activated, the corrective controller $K(x)$ tries to control the system back to the upright configuration.
For the SRL framework, the probability threshold is set to $p_t =0.8$.
Each learning condition is trained with three different seeds, and the averaged results are presented.
The parameters used in the PPO algorithm are given in Appendix C. 

\begin{figure*}[!t]
\sf
\small
\centering
\subfloat[]{\includegraphics[width=.27\linewidth, height= 34mm]{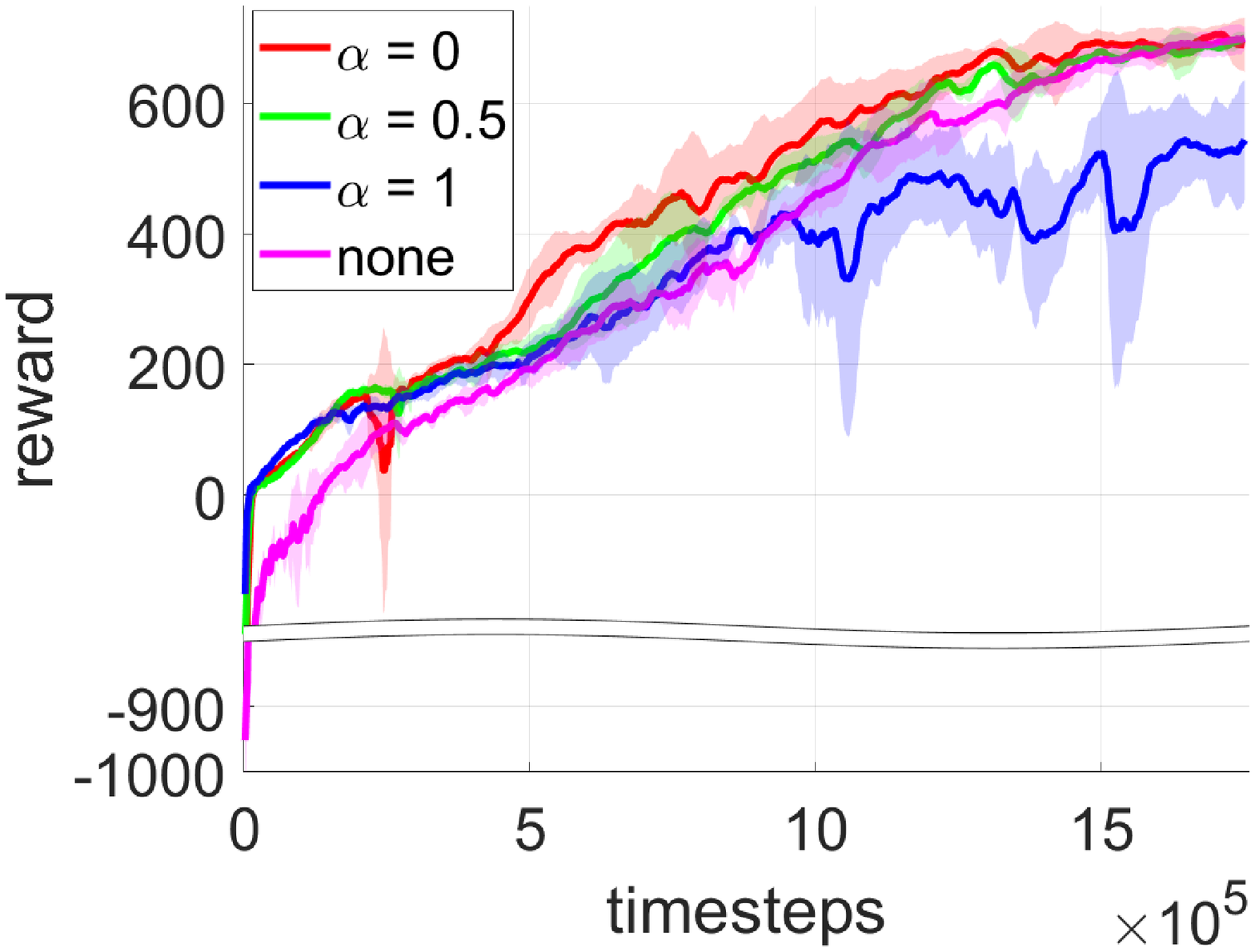}%
\label{fig.learning_01}}
\hfil
\subfloat[]{\includegraphics[width=.27\linewidth, height= 34mm]{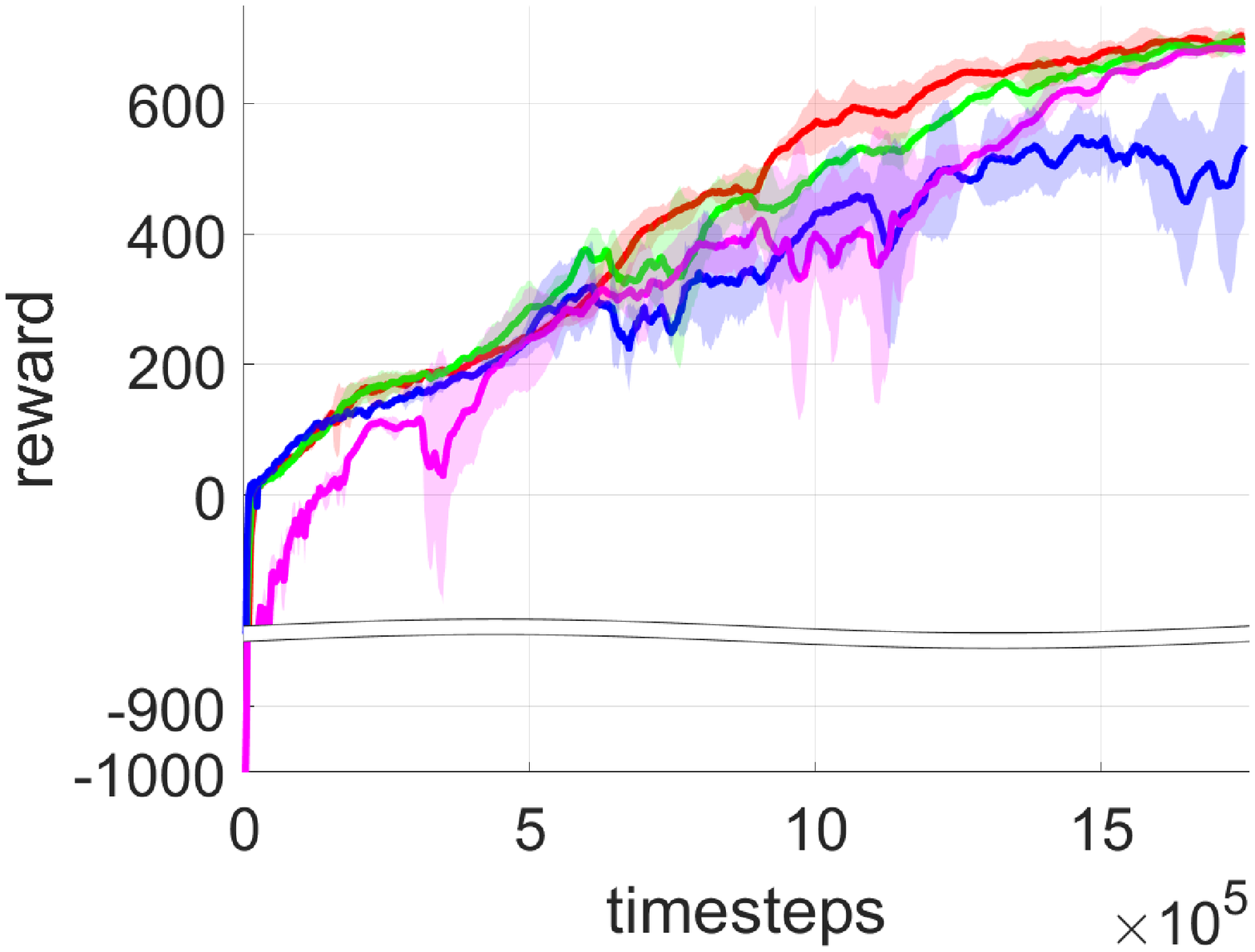}%
\label{fig.learning_05}}
\hfil
\subfloat[]{\includegraphics[width=.27\linewidth, height= 34mm]{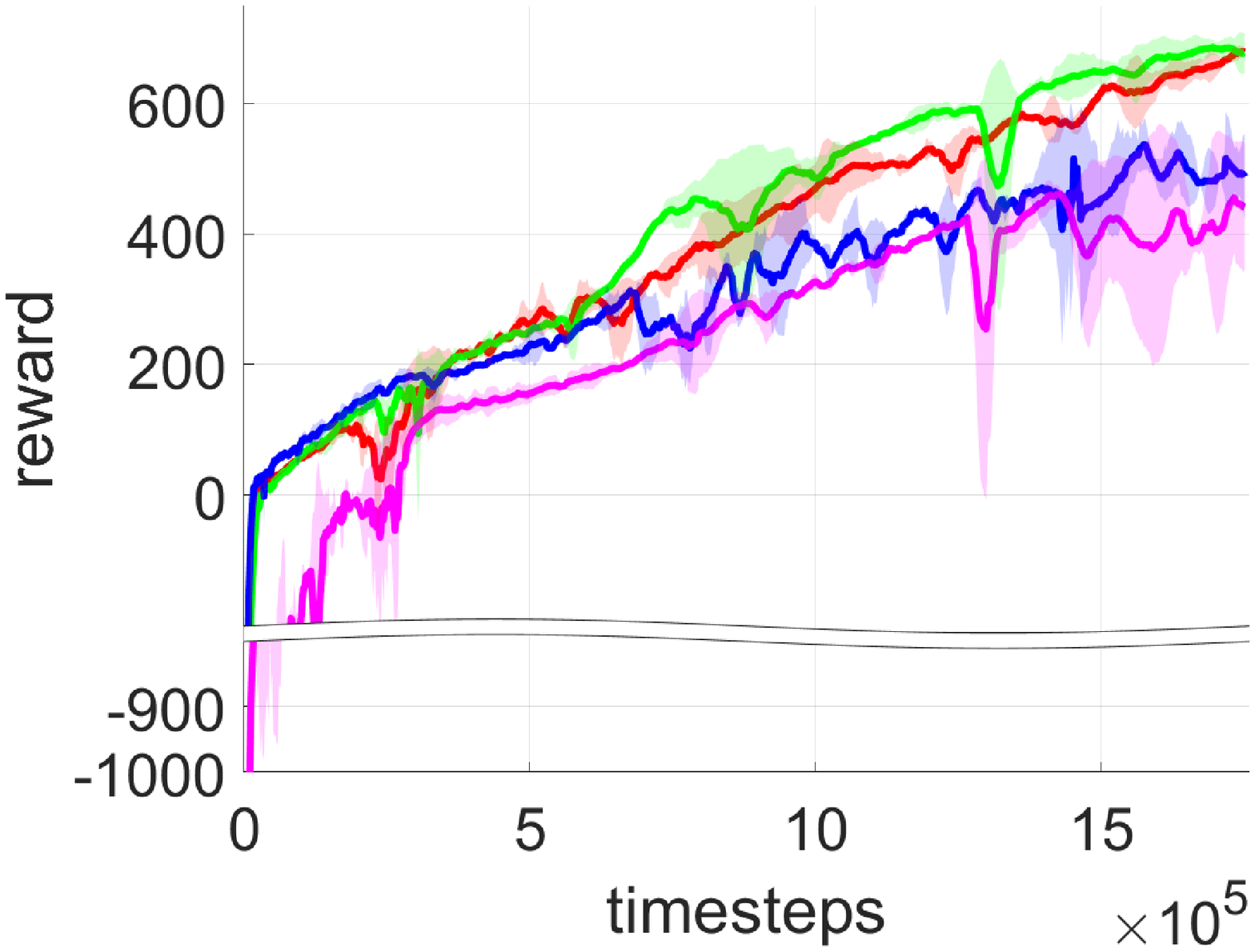}%
\label{fig.learning_4}}
\hfil
\subfloat[]{\includegraphics[width=.27\linewidth, height= 34mm]{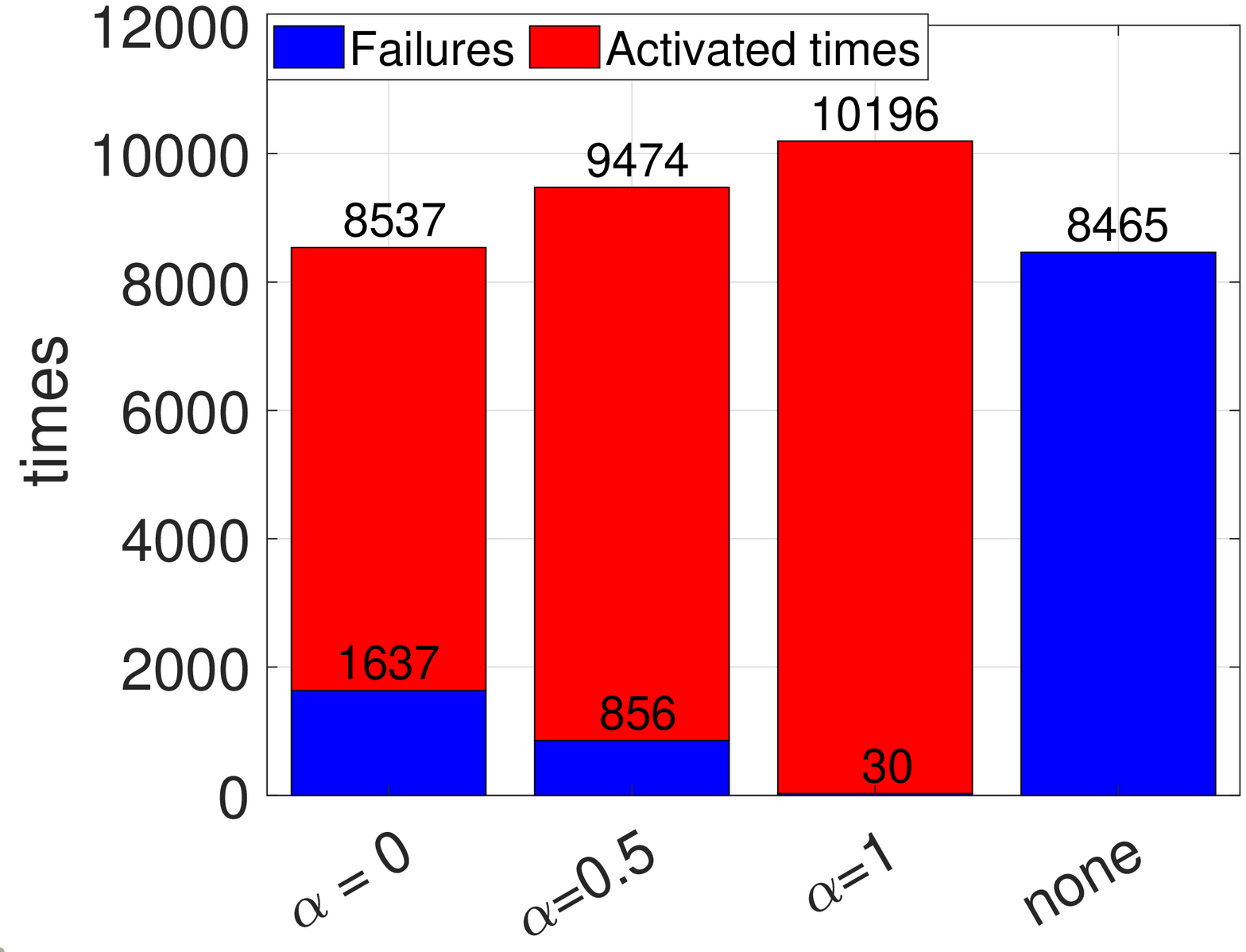}%
\label{fig.learning_failures_01}}
\hfil
\subfloat[]{\includegraphics[width=.27\linewidth, height= 34mm]{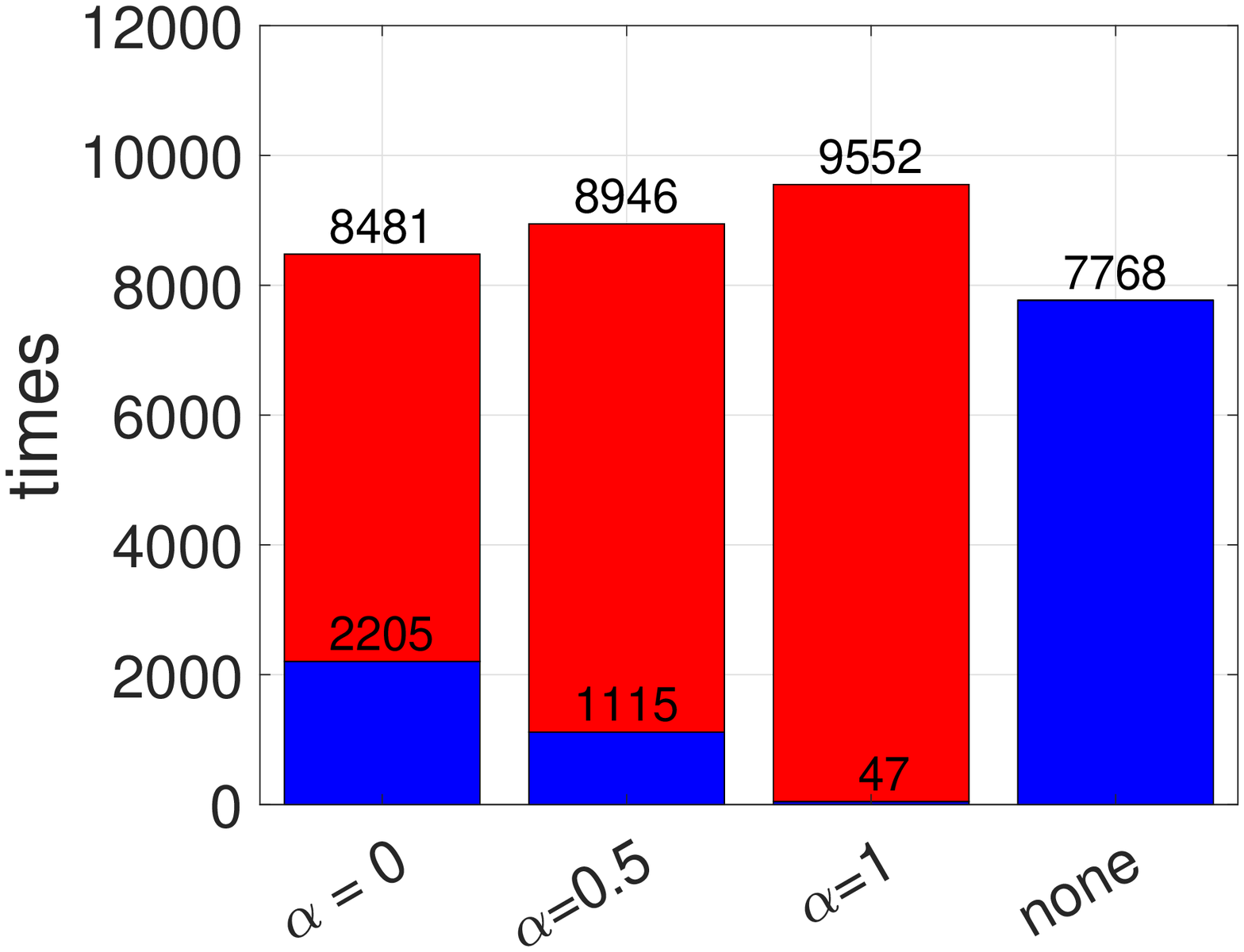}%
\label{fig.learning_failures_05}}
\hfil
\subfloat[]{\includegraphics[width=.27\linewidth, height= 34mm]{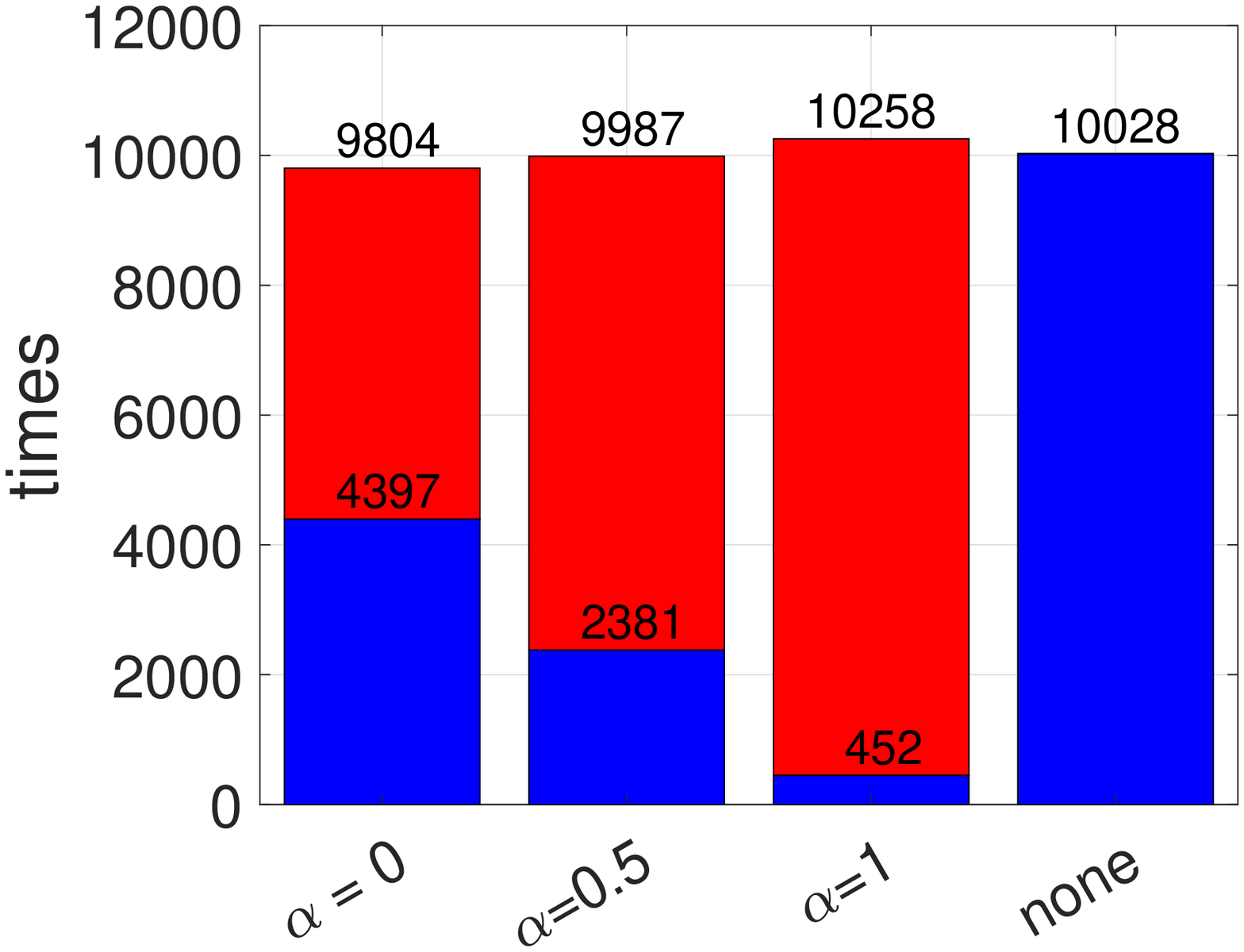}%
\label{fig.learning_failures_4}}
\caption{(a)-(c) Learning performance of the SRL framework with different simplified safe regions obtained with $\alpha =0$, $\alpha =0.5$, $\alpha = 1$ as well as the case that no supervisor is implemented (none), for $\Delta = 1.1$, $\Delta = 1.5$ and $\Delta = 4$, respectively.
(d)-(f) The total number of times that the corrective controller $K(x)$ is activated and the corresponding number of failures for $\Delta = 1.1$, $\Delta = 1.5$ and $\Delta = 4$, respectively.}
\label{fig.learning}
\end{figure*}

\subsection{Estimate of the Simplified Safe Region}
\label{sec.result_initial}

We first examine the influence of the parameter $\alpha$ on the learned simplified safe region $\mathcal{S}_y$. 
We consider the training dataset $\mathbb{D}_{\mathrm{tr}}$ with size $k=1000$ and generate it with three different values of $\alpha$: $\alpha=0$ (using only the MND $\mathcal{D}_{\mathrm{mnd}}$), $\alpha = 0.5$, and $\alpha = 1$ (using only the UD $\mathcal{D}_{\mathrm{ud}}$).
The corresponding results are presented in Fig.~\ref{fig.initial_estimate_results}.

Fig.~\ref{fig.data_gaussian}-\ref{fig.data_uniform} show the distributions of sampled system states contained in the generated training dataset $\mathbb{D}_{\mathrm{tr}}$ by displaying the joint angle and angular velocity of the first link.
As the learning starts at the origin, it is more likely to observe a system state in the neighbourhood of the origin.
Hence, the system states generated from the MND $\mathcal{D}_{\mathrm{mnd}}$ are more dense in the subregions near the origin (see Fig.~\ref{fig.data_gaussian}).
As a result, the proportion of safe data points increases, since the closer to the origin, the higher the probability that a system state can be controlled back to the upright position.
Moreover, limited by the natural dynamics, there exists no feasible control sequence to control the system to a state that simultaneously has a large positive angle and a large negative angular velocity (right-bottom of Fig.~\ref{fig.data_gaussian}), or a large negative angle and a large positive angular velocity (left-top of Fig.~\ref{fig.data_gaussian}) of the first link.
Therefore, no system states are sampled in those subregions.
For the UD $\mathcal{D}_{\mathrm{ud}}$, the sampled system states are placed among the entire state space, and as a consequence, a smaller proportion of safe data points is obtained (see Fig.~\ref{fig.data_uniform}).

Fig.~\ref{fig.y_gaussian}-\ref{fig.data_uniform} are the realizations of simplified states $\{y_1,\ldots, y_k\}$ derived by using t-SNE.
The safe and unsafe training data points are clearly separated in the simplified state space.
The corresponding learned simplified safe region $\mathcal{S}_y$ are presented in Fig.~\ref{fig.gamma_gaussian}-\ref{fig.gamma_uni}.
With less observed safe training data points, the initial estimate of the simplified safe region $\mathcal{S}_y$ becomes more conservative and tends to make an unsafe prediction.
Note that, the axes $y_1$ and $y_2$ obtained with different training datasets have different meanings and cannot be directly compared, as they are the outputs of different t-SNE computations.

\subsection{Data Generation vs. Learning Performance}

We then examine the influence of simplified safe region $\mathcal{S}_y$ learned from different training datasets $\mathbb{D}_{\mathrm{tr}}$ on the performance of the SRL framework.
We compare the performance with three different levels of discrepancy between the nominal and the real systems as $\Delta = 1.1$, $\Delta = 1.5$ and $\Delta = 4$.
Note that, since in this work we focus only on the training data generation, we use the initially learned hypothesis $h(x)$ throughout the entire learning process, such that its influence is better illustrated.
As a baseline for comparisons, we also investigate the learning performance of implementing the PPO directly without the supervisor.
The corresponding results are shown in Fig.~\ref{fig.learning}.

As illustrated in Fig.~\ref{fig.learning_01}-\ref{fig.learning_4}, for all three levels of discrepancy, the training dataset that uses only the UD, i.e., $\alpha = 1$, results in a final policy with a lower reward, since the learning process is overly restricted.
For training dataset that uses only the MND ($\alpha = 0$) or that combines the UD and the MND ($\alpha = 0.5$), the SRL framework is able to find a satisfying final policy.
It is also worth noting that, as each learning trial is terminated when an unsafe behaviour is predicted to occur, an early-stop functionality is introduced to the learning process by using the simplified safe region, which then helps with searching for an optimal policy. 
This effect becomes more significant when the system is hard to control due to heavier masses ($\Delta = 4$), where compared to the free learning case, a better final policy is found when using the SRL framework.

Fig.~\ref{fig.learning_failures_01}-\ref{fig.learning_failures_4} show the total number of times that the corrective controller $K(x)$ is activated during the entire learning process as well as the corresponding number of failures among these safety recoveries.
As a comparison, the total number of times that the safety constraint is violated during the free learning case is also given.
When the discrepancy is small, i.e., $\Delta = 1.1$, the MND $\mathcal{D}_{\mathrm{mnd}}$ results in a success rate of the corrective controller $K(x)$ ($80.8\%$) that is close to the probability threshold $p_t = 0.8$.
While increasing the value of $\alpha$ makes the learning process more conservative, it also ensures a higher probability that the system is safe, e.g. $91.0\%$ for $\alpha = 0.5$ and $99.7\%$ for $\alpha = 1$.
Similar behaviours can also be observed for $\Delta =1.5$ and $\Delta = 4$, though the success rate of the corrective controller $K(x)$ decreases according to the level of the discrepancy.
In general, the supervisor learned from the UD $\mathcal{D}_{\mathrm{ud}}$ is more robust to the mismatches due to its higher conservativeness.


\section{Conclusion}
\label{sec.conclusion}
In this work, we propose a data generation method that is able to provide representative training data for increasing the performance of a SRL framework.
The method divides the training dataset into two parts and use a multivariate normal distribution and a uniform distribution to generate the sub-datasets, respectively.
By adjusting the sizes of the two sub-datasets, a balance between finding a satisfying policy and keeping the system safe is achieved.
The proposed data generation method gives an insight about how different training data will affect the reliability of the safety estimates made via data-driven methods.
For future work, we intend to find a metric for quantifying the discrepancy between the nominal and the real systems, such that it can be used to guide the training data generation process.

\bibliographystyle{IEEEtran}
\bibliography{ref.bib}

\newpage
\appendices
\section{List of Notations}
\label{sec.appendix_notations}

\begin{table}[h!]
\begin{center}
\begin{tabular}{   p{1cm} | p{7cm}   } 
\hline
$u$ & $m$-dimensional input  \\ 
$x$ & $n$-dimensional system state \\ 
$y$ & $n_y$-dimensional simplified state \\
$z$ & safety label \\
$k$ & number of training data points in $\mathbb{D}_{\mathrm{tr}}$ \\
$p_{t}$ & probability threshold for the SRL framework \\
$d(x)$ & unknown part of the system dynamics \\
$K(x)$ & corrective controller \\
$\pi(x)$ & learning-based controller \\
$\mathcal{R}$ & RoA of the real system under the corrective controller \\
$\mathcal{S}$ & safe region of the real system \\
$\mathcal{S}_n$ & safe region of the nominal system \\
$\mathcal{S}_y$ & simplified safe region \\
$l(x)$ & labeling function of the real system given by $\mathcal{S}$ \\
$l_n(x)$ & labeling function of the nominal system given by $\mathcal{S}_n$ \\
$h(x)$ & hypothesis for predicting the safety label of real system states given by the initial estimate of $\mathcal{S}_y$ \\
$\Psi(x)$ & state mapping \\
$\Gamma(y)$ & safety assessment function \\
$\mathbb{D}_{\mathrm{tr}}$ & training dataset \\
$\mathbb{D}_{\mathrm{ud}}$ & sub-dataset generated by using $\mathcal{D}_{\mathrm{ud}}$ \\
$\mathbb{D}_{\mathrm{mnd}}$ & sub-dataset generated by using $\mathcal{D}_{\mathrm{mnd}}$\\
$\mathbb{X}$ & dataset of the observed system states for learning $\mathcal{D}_{\mathrm{mnd}}$ \\
$\mathcal{D}$ & distribution of system states of the real system \\
$\mathcal{D}_n$ & distribution of system states of the nominal system \\
$\mathcal{D}_{\mathrm{ud}}$ & uniform distribution among the system state space \\
$\mathcal{D}_{\mathrm{mnd}}$ & multivariate normal distribution of system states for approximating $\mathcal{D}$ \\ 

\hline
\end{tabular}
\end{center}
\end{table}

\section{Proof of Theorem~\ref{theo.classification_erro}}
\label{sec.appendix_proof}
Let $\epsilon_n (h, l) = \mathrm{E}_{x\sim \mathcal{D}_n} \left[ \mathbb{I} (h(x) \neq l(x)) \right]$, we have
\begin{eqnarray}
    \epsilon (h, l) &=& \epsilon (h, l) + \epsilon_n (h, l_n) - \epsilon_n (h, l_n)  \nonumber \\ 
    && + \epsilon_n (h, l) - \epsilon_n (h, l) \nonumber \\ 
    &\leq& \epsilon_n (h, l_n) + |\epsilon_n (h, l) - \epsilon_n (h, l_n)| \nonumber \\
    && + |\epsilon (h, l) - \epsilon_n (h, l)| \nonumber \\
    &\leq& \epsilon_n (h, l_n) + \mathrm{E}_{x\sim \mathcal{D}_n} \left[ \mathbb{I} (l(x) \neq l_n(x)) \right] \nonumber \\
    && + |\epsilon (h, l) - \epsilon_n (h, l)|
\label{eq.proof_first}
\end{eqnarray}

According to Lemma 3 in~\cite{ben2010theory}, the following holds for any two hypothesis $h_1$ and $h_2$
\begin{equation}
    |\epsilon (h_1, h_2) - \epsilon_n (h_1, h_2) | \leq \frac{1}{2} d_{\mathcal{H}\Delta \mathcal{H}}(\mathcal{D},\mathcal{D}_n)
\end{equation}
By considering the labeling function $l(x)$ as a hypothesis, \eqref{eq.proof_first} becomes
\begin{equation}
     \epsilon (h, l) \leq \epsilon_n (h, l_n) + \mathrm{E}_{x\sim \mathcal{D}_n} \left[ \mathbb{I} (l(x) \neq l_n(x)) \right] + \frac{1}{2}d_{\mathcal{H}\Delta \mathcal{H}}(\mathcal{D},\mathcal{D}_n)
\label{eq.proof_bound_1}
\end{equation}
If in the first line we use $\epsilon (h, l_n) = \mathrm{E}_{x\sim \mathcal{D}} \left[ \mathbb{I} (h(x) \neq l_n(x)) \right]$ instead of $\epsilon_n (h, l)$, we have
\begin{equation}
     \epsilon (h, l) \leq \epsilon_n (h, l_n) + \mathrm{E}_{x\sim \mathcal{D}} \left[ \mathbb{I} (l(x) \neq l_n(x)) \right] + \frac{1}{2}d_{\mathcal{H}\Delta \mathcal{H}}(\mathcal{D},\mathcal{D}_n)
     \label{eq.proof_bound_2}
\end{equation}
Combining~\eqref{eq.proof_bound_1} and~\eqref{eq.proof_bound_2} hence gives the Theorem~\ref{theo.classification_erro}.

\section{Experimental Setup}
\label{sec.appendix_experiment}

For generating the training dataset $\mathbb{D}_{\mathrm{tr}}$, the sample ranges of each state variable used in the UD are chosen as: $-\frac{\pi}{2}\si{\radian} < \theta_1 < \frac{\pi}{2}\si{\radian}$, $-\pi \si{\radian}\leq \theta_2,\theta_3 \leq \pi \si{\radian}$, $-10 \si{\radian\per\sec} \leq \dot{\theta_1} \leq 10 \si{\radian\per\sec}$, $-20 \si{\radian\per\sec} \leq \dot{\theta_2},\dot{\theta_3} \leq 20 \si{\radian\per\sec}$. 
The following reward function is used for the PPO algorithm
\begin{equation}
    R(t) = R_c - 10 \cdot ||p_e(t)-p_d(t)||
\end{equation}
where $R_c = 2$ is a constant reward for being safe, $||p_e(t)-p_d(t)||$ is the Euclidean distance between the current end-effector position $p_e$ and the desired position $p_d$ given by the trajectory of the target circle.
Following parameters are used for the PPO: the number of steps per update is $2048$, the value function loss coefficient is $1$, the gradient norm clipping coefficient is $10$, the learning rate is $1e^{-4}$, the number of training epochs per update is $10$, the number of training minibatches per update is $128$, the discounting factor is $0.99$, the advantage estimation discounting factor is $0.95$, the policy entropy coefficient is $0$, the number of hidden layers of the neural network is $2$ with $128$ neurons in each layer.

\end{document}